\documentclass[submission,copyright]{eptcs}
\providecommand{\event}{SOS/EXPRESS 2020} %

\title{Substructural Observed Communication Semantics}
\author{Ryan Kavanagh
  \institute{Computer Science Department\\
    Carnegie Mellon University\\
    Pittsburgh, Pennsylvania, 15213-3891, USA}
  \email{rkavanagh@cs.cmu.edu}
}
\def\titlerunning{Substructural Observed Communication Semantics}
\def\authorrunning{R. Kavanagh}

\usepackage[utf8]{inputenc}

\usepackage{etoolbox}
\usepackage{stmaryrd}

\usepackage{bm}
\usepackage[T1]{fontenc}
\usepackage{microtype}

\usepackage{proof}
\usepackage{mathtools}
\usepackage{mathpartir}
\usepackage{amsthm}
\usepackage{amssymb}

\usepackage{hyperref}

\usepackage{enumitem}

\usepackage[canadian]{babel}

\usepackage[autostyle=true]{csquotes}

\usepackage[style=eptcs]{biblatex}
\usepackage{breakurl}

\usepackage{cleveref}
\usepackage{cleveref}
\crefname{diagram}{diagram}{diagrams}
\crefname{property}{property}{properties}
\crefname{intn}{interpretation}{interpretations}
\creflabelformat{intn}{#2(#1)#3}
\crefrangelabelformat{intn}{(#3#1#4) to~(#5#2#6)}

\usepackage[dvipsnames]{xcolor}
\usepackage{color}
\newcommand{\red}[1]{{\color{Maroon}{#1}}}
\newcommand{\blue}[1]{{\color{MidnightBlue}{#1}}}

\usepackage[all,cmtip,2cell]{xy}
\usepackage{tikz}
\usepackage{tikz-cd}
\usepackage{forest}
\tikzcdset{
  arrow style=tikz,
  diagrams={>={Computer Modern Rightarrow[scale=0.8]}},
  diagrams={>={To}},
}

\usepackage{booktabs}
\usepackage{adjustbox}

\usepackage[draft]{fixme}
\fxsetup{theme=color}
\FXRegisterAuthor{rk}{anrk}{RAK}
\FXRegisterAuthor{fp}{anfp}{FP}
\FXRegisterAuthor{sb}{ansb}{SB}

\usepackage{xspace}
\newcommand{\ie}{i.e.\@\xspace}
\newcommand{\eg}{e.g.\@\xspace}

\makeatletter
\newcommand*{\etc}{%
  \@ifnextchar{.}%
  {etc}%
  {etc.\@\xspace}%
}
\makeatother

\newcommand{\ms}{\mathsf}
\newcommand{\mb}{\mathbf}

\newcommand{\mc}{\mathcal}

\newcommand{\defin}[1]{{\usefont{T1}{lmss}{b}{n}{#1}}}

\usepackage{stackengine,scalerel}
\newcommand\operatorupX[1]{\,\ThisStyle{\ensurestackMath{%
      #1\stackengine{-0pt}{\,}{\SavedStyle\!^{\mathord{\uparrow}}}{O}{l}{F}{T}{S}}}\!}
\makeatletter
\newcommand\operatorup[1]{\!\mathop{\operatorupX{#1}}\@ifnextchar\{{\,}{\@ifnextchar[{\,}{\@ifnextchar({\,}{}}}}
\makeatother

\newcommand{\jtype}[4]{{#1} \vdash {#2} \mathrel{::} {#3} : {#4}}
\newcommand{\jtypem}[5]{{#1}\mathrel{;} {#2} \vdash {#3} \mathrel{{:}{:}} {#4} : {#5}}

\newcommand{\jstep}[3][{}]{#2 \xrightarrow{} #3}

\newcommand{\jproc}[2]{\mathsf{proc}(#1, #2)}

\newcommand{\jmsg}[2]{\mathsf{msg}(#1, #2)}

\newcommand{\outm}[1]{\red{#1}}
\newcommand{\inm}[1]{\blue{#1}}

\newcommand{\jttp}[3]{\inm{#1} \vdash \inm{#2} : \outm{#3}}
\newcommand{\jsynt}[2]{\inm{#1} \mathrel{\varepsilon} \inm{#2}}
\newcommand{\jtoc}[4]{\inm{#1} \leadsto \outm{#2} \mathrel{\varepsilon} \outm{#4} \mathrel{/} \inm{#3}}
\newcommand{\jctype}[2]{{\ensuremath \mb{type}(#1 : #2)}}

\newcommand{\step}{\longmapsto}

\newcommand{\tSendC}[3]{\mathsf{send}\ #1\ #2;\ #3}
\newcommand{\tSendL}[3]{#1.#2;#3}

\newcommand{\tSendU}[2]{\ms{send}\ #1\ \mathsf{unfold}; #2}
\newcommand{\tRecvC}[3]{#1 \leftarrow \mathsf{recv}\ #2;\ #3}

\newcommand{\tRecvU}[2]{\ms{unfold} \leftarrow \ms{recv}\ #1; #2}
\newcommand{\tClose}[1]{\mathsf{close}\ #1}
\newcommand{\tWait}[2]{\mathsf{wait}\ #1;#2}
\newcommand{\tCase}[2]{\mathsf{case}\ #1\ #2}
\newcommand{\tFwd}[2]{#1 \leftarrow #2}
\newcommand{\tFwdP}[2]{#1 \rightarrow #2}

\newcommand{\tCut}[3]{#1 \leftarrow #2;\ #3}
\newcommand{\tFix}[2]{\ms{fix}\ #1.#2}

\newcommand{\mSendL}[3]{\tSendL{#1}{#2}{\tFwd{#1}{#3}}}
\newcommand{\mSendC}[3]{\tSendC{#1}{#2}{\tFwd{#1}{#3}}}

\newcommand{\mSendU}[2]{\tSendU{#1}{\tFwd{#1}{#2}}}

\newcommand{\Tu}{\mathbf{1}}
\newcommand{\Tplus}{\oplus}
\newcommand{\Tot}{\otimes}
\newcommand{\Tamp}{\&}

\newcommand{\Trec}[2]{\rho #1.#2}

\newcommand{\Tproc}[2]{\{#1 \leftarrow #2\}}

\newcommand{\C}{\mathcal{C}}
\newcommand{\N}{\mathbb{N}}

\DeclareMathOperator{\inpc}{ic}
\DeclareMathOperator{\outc}{oc}

\DeclareMathOperator{\fcn}{fc}

\DeclareMathOperator{\gfp}{gfp}

\newcommand{\obsbr}[1]{\llparenthesis {#1} \rrparenthesis}

\newtheorem{theorem}{Theorem}

\crefname{conjecture}{conjecture}{conjectures}

\newtheorem{example}{Example}
\newtheorem{proposition}{Proposition}
\newtheorem{lemma}{Lemma}
\newtheorem{corollary}{Corollary}

\newcommand{\rn}[1]{(\textsc{#1})}

\newcommand{\M}{\textbf{M}}

\newcommand{\gact}[2]{#1 \mathbin{\cdot} #2}
\newcommand{\olap}[2]{\Omega_{#1}(#2)}
\newcommand{\card}[1]{\mb{#1}}
\newcommand{\unfold}{\ms{unfold}}

\DeclareMathOperator{\queue}{\ms{que}}%
\DeclareMathOperator{\enq}{\ms{enq}}%
\DeclareMathOperator{\supp}{supp}

\begin{document}

\maketitle

\begin{abstract}
  Session-types specify communication protocols for communicating processes, and session-typed languages are often specified using substructural operational semantics given by multiset rewriting systems.
  We give an observed communication semantics~\cite{atkey_2017:_obser_commun_seman_class_proces} for a session-typed language with recursion, where a process's observation is given by its external communications.
  To do so, we introduce \textit{fair executions} for multiset rewriting systems, and extract observed communications from fair process executions.
  This semantics induces an intuitively reasonable notion of observational equivalence that we conjecture coincides with semantic equivalences induced by denotational semantics~\cite{kavanagh_2020:_domain_seman_higher_v3}, bisimulations~\cite{gommerstadt_2018:_session_typed_concur_contr}, and barbed congruences~\cite{toninho_2015:_logic_found_session_concur_comput, kokke_2019:_better_late_than_never} for these languages.
\end{abstract}

\section{Introduction}
\label{sec:introduction}

A proofs-as-processes correspondence between linear logic and the session-typed $\pi$-calculus is the basis of many programming languages for message-passing concurrency~\cite{caires_pfenning_2010:_session_types_intuit_linear_propos,caires_2016:_linear_logic_propos,toninho_2011:_depen_session_types,wadler_2014:_propos_as_session}.
Session types specify communication protocols, and all communication with session-typed processes must respect these protocols.
If we take seriously the idea that we can only interact with processes through session-typed communication, then the only thing we can observe about them is their communications.
Indeed, timing differences in communication are not meaningful due to the non-deterministic scheduling of process reductions, and ``forwarding'' or ``linking'' of channels renders process termination meaningless, even in the presence of recursion.
It follows that processes should be observationally indistinguishable only if they always send the same output given the same input.

These ideas underlie Atkey's~\cite{atkey_2017:_obser_commun_seman_class_proces} novel \textit{observed communication semantics} (OCS) for Wadler's Classical Processes~\cite{wadler_2014:_propos_as_session}.
Atkey's OCS uses a big-step evaluation semantics to observe communications on channels deemed ``observable''.
Processes are then observationally equivalent whenever they have the same observed communications in all contexts.

Building on these ideas, \textbf{we give an OCS for session-typed languages that are specified using substructural operational semantics} (SSOS), a form of multiset rewriting.
Our work differs from Atkey's on several key points.
First, we assume that communication is asynchronous rather than synchronous.
This assumption costs us nothing, for synchronous communication can be encoded in asynchronous systems~\cite{pfenning_griffith_2015:_polar_subst_session_types}, and it simplifies the semantics by eliminating the need for ``configurations'' and ``visible'' cuts.
More importantly, \textbf{our OCS supports recursive and non-terminating processes}.
To do so, we observe communications from process traces in (a conservative extension of) the usual SSOSs, instead of defining a separate big-step semantics.

To ensure that observed communications are well-defined in the presence of non-termination, we require that process executions be \textit{fair}.
Intuitively, fairness ensures that if a process can make progress, then it eventually does so.
Fairness is also motivated by ongoing efforts to relate existing SSOSs to domain-theoretic semantics for this style of language~\cite{kavanagh_2020:_domain_seman_higher_v3}.
There, processes denote continuous functions between domains of session-typed communications, and fairness is built-in.
To this end, \textbf{we introduce \textit{fair executions} of multiset rewriting systems} (MRS) and give sufficient conditions for an MRS to have fair executions.
We also introduce a new notion of trace equivalence, \textit{union-equivalence}, that is key to defining our OCS.

We study fair executions of MRSs and their properties in \cref{sec:fair-mult-rewr}.
In \cref{sec:sess-typed-proc}, we give an SSOS for a session-typed language arising from a proofs-as-processes interpretation of intuitionistic linear logic.
It supports recursive processes and types.
Though it is limited, it represents the core of other SSOS-specified session-typed languages~\cite{balzer_pfenning_2017:_manif_sharin_with_session_types,gommerstadt_2018:_session_typed_concur_contr,kavanagh_2020:_domain_seman_higher_v3,pfenning_griffith_2015:_polar_subst_session_types,toninho_2013:_higher_order_proces_funct_session}, and the techniques presented in this paper scale to their richer settings.
In \cref{sec:observ-comm}, we give our observed communication semantics, where we use a coinductively defined judgment to extract observations from fair executions.

\section{Fair Executions of Multiset Rewriting Systems}
\label{sec:fair-mult-rewr}

In this section, we introduce \textit{fairness} and \textit{fair executions} for multiset rewriting systems.
We begin by revisiting (first-order) multiset rewriting systems, as presented by Cervesato et~al.~\cite{cervesato_2005:_compar_between_stran}.
We present a notion of fairness for sequences of rewriting steps, and constructively show that under reasonable hypotheses, all fair sequences from the same multiset are permutations of each other.
We introduce a new notion of trace equivalence, ``union-equivalence'', and give sufficient conditions for traces to be union-equivalent.
Fairness and union-equivalence will be key ingredients for defining the observed communication semantics of \cref{sec:observ-comm}.

A \defin{multiset} $M$ is a pair $(S, m)$ where $S$ is a set (the \textit{underlying set}) and $m : S \to \N$ is a function.
It is finite if $\sum_{s \in S} m(s)$ is finite.
We say $s$ is an \defin{element} of $M$, $s \in M$, if $m(s) > 0$.
When considering several multisets, we assume without loss of generality that they have equal underlying sets.
The \defin{sum} $M_1,M_2$ of multisets $M_1 = (S, m_1)$ and $M_2 = (S, m_2)$ is the multiset $(S, \lambda s \in S.m_1(s) + m_2(s))$.
Their \defin{intersection} $M_1 \cap M_2$ is the multiset $(S, \lambda s \in S . \min(m_1(s), m_2(s)))$.
Their \defin{difference} $M_1 \setminus M_2$ is the multiset $(S, \lambda s \in S . \max(0, m_1(s) - m_2(s)))$.
We say that $M_1$ is \defin{included} in $M_2$, written $M_1 \subseteq M_2$, if $m_1(s) \leq m_2(s)$ for all $s \in S$.

Consider finite multisets $M$ of first-order atomic formulas over some signature whose constants are drawn from some countably infinite set.
We call closed formulas \defin{judgments}.
Judgments represent facts, some of which we may deem to be persistent.
To this end, we partition formulas as \defin{persistent} (indicated by bold face, $\mb{p}$) and \defin{ephemeral} (indicated by sans serif face, $\ms{p}$).
We write $M(\vec x)$ to mean that the formulas in $M$ draw their variables from $\vec x$.
A \defin{multiset rewrite rule} $r$ is an ordered pair of multisets $F(\vec x)$ and $G(\vec x, \vec n)$, where the multiset $\pi(\vec x)$ of persistent formulas in $F(\vec x)$ is included in $G(\vec x, \vec n)$.
We interpret the variables $\vec x$ as being universally quantified and the variables $\vec n$ as being existentially quantified.
This relation is made explicit using the syntax
\[
  r : \forall \vec x . F(\vec x) \to \exists \vec n . G(\vec x, \vec n).
\]
In practice, we often elide $\forall \vec x$ and do not repeat the persistent formulas $\pi(\vec x) \subseteq F(\vec x)$ on the right side of the arrow.
A \defin{multiset rewriting system} (MRS) is a set $\mc{R}$ of multiset rewrite rules.

Multiset rewrite rules describe localized changes to multisets of judgments.
Given a rule $r : \forall \vec x . F(\vec x) \to \exists \vec n . G(\vec x, \vec n)$ in $\mc{R}$ and some choice of constants $\vec c$ for $\vec x$, we say that the \defin{instantiation} $r(\vec c) : F(\vec c) \to \exists \vec n.G(\vec c, \vec n)$ is \defin{applicable} to a multiset $M$ of judgments if there exists a multiset $M'$ such that $M = F(\vec c),M'$.
The rule $r$ is applicable to $M$ if $r(\vec c)$ is applicable to $M$ for some $\vec c$.
In these cases, the \defin{result} of applying $r(\vec c)$ to $M$ is the multiset $G(\vec c, \vec d),M'$, where $\vec d$ is a choice of fresh constants.
In particular, we assume that the constants $\vec d$ do not appear in $M$ or in $\mc{R}$.
We call $\theta = [\vec c / \vec x]$ the \defin{matching substitution} and $\xi = [\vec d / \vec n]$ the \defin{fresh-constant substitution}.
The \defin{instantiating substitution} for $r$ relative to $M$ is the composite substitution $\delta = (\theta, \xi)$.
We capture this relation using the syntax
\[
  F(\vec c),M' \xrightarrow{(r;\delta)} G(\vec c, \vec n),M'.
\]
For conciseness, we often abuse notation and write $r(\theta)$, $F(\theta)$, and $G(\theta,\xi)$ for $r(\vec c)$, $F(\vec c)$, and $G(\vec c, \vec d)$.
We call $F(\vec c)$ the \defin{active} multiset and $M'$ the \defin{stationary} multiset.

Given an MRS $\mc{R}$ and a multiset $M_0$, a \defin{trace} from $M_0$ is a countable sequence of steps
\begin{equation}
  \label{eq:obssem:7}
  M_0 \xrightarrow{(r_1;\delta_1)} M_1 \xrightarrow{(r_2;\delta_2)} M_2 \xrightarrow{(r_3;\delta_3)} \cdots
\end{equation}
such that, where $\delta_i = (\theta_i, \xi_i)$,
\begin{enumerate}
\item for all $i$, $\xi_i$ is one-to-one;
\item for all $i < j$, the constants in $M_i$ and $\xi_j$ are disjoint.
\end{enumerate}
The notation $(M_0, (r_i;\delta_i)_{i \in I})$ abbreviates the trace \eqref{eq:obssem:7}, where $I$ always ranges over $\N^+$ or $\card{n} = \{1, \dotsc, n \}$ for some $n \in \N$.
An \defin{execution} is a maximally long trace.

\begin{example}
  \label[example]{ex:fair-mult-rewr:1}%
  We model queues using an MRS.
  Let the judgment $\queue(q, \$)$ mean that $q$ is the empty queue, and let $\queue(q, v \to q')$ mean that the queue $q$ has value $v$ at its head and that its tail is the queue $q'$.
  Then the multiset $Q = \queue(q, 0 \to q'), \queue(q', \$)$ describes a one-element queue containing $0$.
  The following two rules capture enqueuing values on empty and non-empty queues, respectively, where the formula $\enq(q, v)$ is used to enqueue $v$ onto the queue $q$:
  \begin{align*}
    e_1 &: \forall x, y . \enq(x, y), \queue(x, \$) \to \exists z . \queue(x, y \to z), \queue(z, \$),\\
    e_2 &: \forall x, y, z, w . \enq(x, y), \queue(x, z \to w) \to \queue(x, z \to w), \enq(w, y).
  \end{align*}
  The following sequence is an execution from $Q,\enq(q,1)$, and it captures enqueuing 1 on the queue $q$:
  \[
    Q,\enq(q,1) \xrightarrow{(e_2;([q,1,0,q'/x,y,z,w], \emptyset))} Q,\enq(q',1) \xrightarrow{(e_1;([q',1/x,y], [a/z]))} \queue(q, 0 \to q'), \queue(q', 1 \to a), \queue(a, \$).
  \]
\end{example}

The constants in fresh-constant substitutions are not semantically meaningful, so we identify traces up to refreshing substitutions.
A \defin{refreshing substitution} for a trace $T = (M_0,(r_i;(\theta_i,\xi_i))_i)$ is a collection of fresh-constant substitutions $\eta = (\eta_i)_i$ such that $[\eta]T = (M_0, (r_i; (\theta_i, \eta_i))_i)$ is also a trace.
Explicitly, we identify traces $T$ and $T'$ if there exists a refreshing substitution $\eta$ such that $T' = [\eta]T$.

Given rules $r_i : \forall \vec x_i . F_i(\vec x_i) \to \exists \vec n_i . G_i(\vec x_i, \vec n_i)$ and matching substitutions $\theta_i$ for $i = 1, 2$, we say that the instantiations $r_1(\theta_1)$ and $r_2(\theta_2)$ are \defin{equivalent}, $r_1(\theta_1) \equiv r_2(\theta_2)$, if both $F_1(\theta_1) = F_2(\theta_2)$ and (up to renaming of bound variables) $\exists \vec n_1 . G_1(\theta_1, \vec n_1) = \exists \vec n_2 . G_2(\theta_2, \vec n_2)$; otherwise they are \defin{distinct}.
Application does not distinguish between equivalent instantiations: if $r_1(\theta_1) \equiv r_2(\theta_2)$ are applicable to $M_0$, then applying each to $M_0$ gives the same result up to refreshing substitution.

Given an MRS $\mc{R}$, we say that an execution $(M_0,(r_i;\delta_i)_{i \in I})$ is \defin{fair} if for all $i \in I$, $r \in \mc{R}$, and $\theta$, whenever $r(\theta)$ is applicable to $M_i$, there exists a $j > i$ such that $r_j(\theta_j) \equiv r(\theta)$.
Given a fair trace $T$, we write $\phi_T(i,r,\theta)$ for the least such $j$.
In the case of MRSs specifying SSOSs of session-typed languages, this notion of fairness implies \textit{strong process fairness}~\cite{francez_1986:_fairn, park_1982:_predic_trans_weak_fair_iterat, costa_stirling_1987:_weak_stron_fairn_ccs}, which guarantees that if a process can take a step infinitely often, then it does so infinitely often.
In particular, it implies that if a process can take a step, then it eventually does so.

\begin{example}
  \label[example]{ex:fair-mult-rewr:3}
  The execution of \cref{ex:fair-mult-rewr:1} is fair.
\end{example}

\begin{proposition}[Fair Tail Property]
  \label[proposition]{prop:main:4}
  If $(M_0, (t_i;\delta_i)_{i \in I})$ is fair, then so is $(M_n, (t_i;\delta_i)_{\substack{n<i, i \in I}})$ for all $n \in I$.
\end{proposition}

We consider various criteria that imply fairness.
The first will be interference-freedom, which roughly means that at any given point, the order in which we apply applicable rules does not matter.
It will hold whenever the rules do not ``overlap''.
In general, given an MRS $\mc{R}$ and a property $P$, we say $P$ holds \defin{from} $M_0$ if for all traces $(M_0, (r_i;\delta_i)_{i \in I})$, $P$ holds for $M_0$ and for $M_i$ for all $i \in I$.

Write $S_I$ for the group of bijections on $I$; its elements are called permutations.
A permutation $\sigma \in S_I$ acts on a trace $T = (M_0, (t_i;\delta_i)_{i \in I})$ to produce a sequence $\gact{\sigma}{T} = (M_0, (t_{\sigma(i)};\delta_{\sigma(i)})_{i \in I})$.
This sequence $\gact{\sigma}{T}$ is a \defin{permutation of} $T$ whenever it is also a trace.
We adopt group-theoretic notation for cyclic permutations and write $(x, \sigma(x), \sigma(\sigma(x)), \dotsc)$ for a cyclic permutation $\sigma : I \to I$; implicit is that all elements not in the orbit of $x$ are fixed by $\sigma$.
Cycles of length two are called transpositions.

Consider an MRS $\mc{R}$ and let $r_1(\theta_1), \dotsc, r_n(\theta_n)$ enumerate all distinct instantiations of rules in $\mc{R}$ applicable to $M_0$.
We say that $\mc{R}$ \defin{commutes} on $M_0$ or is \defin{interference-free} on $M_0$ if for all corresponding pairwise-disjoint fresh-constant substitutions $\xi_i$, the following diagram commutes for all permutations $\sigma \in S_{\card{n}}$, and both paths around it are traces:
\[
  \begin{tikzcd}[column sep=6em, row sep=1em]
    &
    M_1
    \ar[r, "{(r_2;(\theta_2,\xi_2))}"]
    &[1ex]
    \cdots
    \ar[r, "{(r_{n - 1};(\theta_{n - 1},\xi_{n - 1}))}"]
    &[2.5em]
    M_{n - 1}
    \ar[dr, "{(r_n;(\theta_n,\xi_n))}" pos=0.3]
    &
    \\
    M_0
    \ar[ur, "{(r_1;(\theta_1,\xi_1))}"  pos=0.7]
    \ar[dr, swap, "{(r_{\sigma(1)};(\theta_{\sigma(1),\xi_\sigma(1)}))}" pos=0.7]
    &
    &
    &
    &
    M_n
    \\
    &
    M_1'
    \ar[r, "{(r_{\sigma(2)};(\theta_{\sigma(2)},\xi_{\sigma(2)}))}"]
    &
    \cdots
    \ar[r, "{(r_{\sigma(n - 1)};(\theta_{\sigma(n - 1)},\xi_{\sigma(n - 1)}))}"]
    &
    M_{n - 1}'
    \ar[ur, swap, "{(r_{\sigma(n)};(\theta_{\sigma(n)},\xi_{\sigma(n)}))}" pos=0.3]
    &
  \end{tikzcd}
\]
We note that interference-freedom is only defined if the enumeration of distinct applicable instantiations is finite.
The following proposition is an immediate consequence of the definition of commuting rules:

\begin{proposition}
  \label[proposition]{prop:fair-mult-rewr:3}
  Let $\mc{R}$ commute on $M_0$, and let $r_i(\theta_i)$ with $1 \leq i \leq n$ be the distinct instantiations applicable on $M_0$.
  If $M_0 \xrightarrow{(r_1;(\theta_1,\xi_1))} M_1$, then $r_2(\theta_2), \dotsc, r_n(\theta_n)$ are applicable to and commute on $M_1$.
\end{proposition}

Interference-freedom implies the existence of fair executions:

\begin{proposition}[Fair Scheduler]
  \label{prop:main:9}
  Assume the axiom of countable choice.
  If $\mc{R}$ is interference-free from $M_0$, then there is a fair execution from $M_0$.
\end{proposition}

\begin{proof}[Proof (Sketch)]
  Let $Q$ be a queue of rule instantiations applicable to $M_0$.
  Given $M_n$, dequeue a rule $r_{n + 1}(\theta_{n + 1})$ from $Q$ and use the axiom of countable choice to choose a suitably disjoint fresh-constant substitution $\xi_{n + 1}$.
  By interference-freedom, it is applicable to $M_n$, and let $M_{n + 1}$ be the result of doing so.
  Enqueue all newly-applicable rule instantiations.
  If $Q$ is ever empty, then the trace is finite but maximally long.
  In all cases, the trace gives a fair execution: every distinct applicable rule instantiation is enqueued and then applied after some finite number of steps.
\end{proof}

Though interference-freedom simplifies fair scheduling, it is primarily of interest for reasoning about executions.
For example, it is useful for showing confluence properties.
It also lets us safely permute certain steps in a trace without affecting observations for session-typed processes (see~\cref{sec:observ-comm}).
This can simplify process equivalence proofs, because it lets us assume that related steps in an execution happen one after another.

Interference-freedom is a strong property, but it arises frequently in nature.
This is because many systems can be captured using rules whose active multisets do not overlap, and rules whose active multisets are non-overlapping commute.
In fact, even if their active multisets overlap, the rules do not disable each other so long as they preserve these overlaps.

To make this intuition explicit, consider multisets $M_i \subseteq M$ for $1 \leq i \leq n$.
Their \defin{overlap} in $M$ is $\olap{M}{M_1, \dotsc, M_n} = M_1,\dotsc,M_n \setminus M$.
Consider an MRS $\mc{R}$ and let $r_i(\theta_i) : F_i(\theta_i) \to \exists \vec n_i . G_i(\theta_i, \vec n_i)$, $1 \leq i \leq n$, enumerate all distinct instantiations of rules in $\mc{R}$ applicable to $M$.
We say that $\mc{R}$ is \defin{non-overlapping} on $M$ if for all $1 \leq i \leq n$ and  fresh-constant substitutions $\xi_i$, $F_i(\theta_i) \cap \olap{M}{F_1(\theta_1), \dotsc, F_n(\theta_n)} \subseteq G_i(\theta_i, \xi_i)$.

\begin{example}
  \label[example]{ex:fair-mult-rewr:2}
  The MRS given by \cref{ex:fair-mult-rewr:1} is non-overlapping from any multiset of the form $Q,E$ where $Q$ is a queue rooted at $q$, and $E$ contains at most one judgment of the form $\enq(q,v)$.
\end{example}

\Cref{prop:fair-mult-rewr:1} characterizes the application of non-overlapping rules, while \cref{prop:main:13} characterizes the relationship between commuting and non-overlapping rules.

\begin{proposition}
  \label[proposition]{prop:fair-mult-rewr:1}
  Let $\mc{R}$ be non-overlapping on $M_0$ and let $r_i(\theta_i) : F_i(\theta_i) \to \exists \vec n_i.G(\theta_i, \vec n_i)$ with $1 \leq i \leq n$ be the distinct instantiations applicable to $M_0$.
  If $M_0 \xrightarrow{(r_1;(\theta_1,\xi_1))} M_1$ and $r_1, \dotsc, r_n$ are non-overlapping on $M_0$, then $r_2(\theta_2), \dotsc, r_n(\theta_n)$ are applicable to and non-overlapping on $M_1$.

  In particular, set $O = \olap{M_0}{F_1, \dotsc, F_n} \cap F_1$.
  There exist $F_1'$ and $G_1'$ be such that $F_1 = O,F_1'$ and $G_1 = O,G_1'$, and there exists an $M$ such that $M_0 = O,F_1',M$ and $M_1 = O,G_1',M$.
  The instantiations $r_2(\theta_2), \dotsc, r_n(\theta_n)$ are all applicable to $O,M \subseteq M_1$.
\end{proposition}

\begin{proposition}
  \label[proposition]{prop:main:13}
  An MRS commutes on $M_0$ if it is non-overlapping on $M_0$; the converse is false.
\end{proposition}

For the remainder of this section, assume that if $(M_0,(r_i;\delta_i)_i)$ is a fair trace, then its MRS is interference-free from $M_0$.
Interference-freedom implies the ability to safely permute finitely many steps that do not depend on each other.
However, it is not obvious that finite permutations, let alone infinite permutations, preserve fairness.
To show that they do, we use the following lemma to reduce arguments about infinite permutations to arguments about finite permutations:

\begin{lemma}
  \label[lemma]{lemma:main:9}
  For all $n \in \N$ and permutations $\sigma : \N \to \N$, set $\chi_\sigma(n) = \sup_{k \leq n} \sigma^{-1}(k)$.
  Then there exist permutations $\tau, \rho : \N \to \N$ such that $\sigma = \rho \circ \tau$, $\tau(k) = k$ for all $k > \chi_\sigma(n)$, and $\rho(k) = k$ for all $k \leq n$.
\end{lemma}

The following proposition shows that permutations of prefixes of traces preserve fairness.
Its proof uses a factorization of permutations into cycles permuting adjacent steps, where each cycle preserves~fairness.

\begin{proposition}
  \label[proposition]{prop:fair-mult-rewr:6}
  Consider an MRS $\mc{R}$ that is interference-free from $M_0$ and let $T = (M_0,(r_i;(\theta_i,\xi_i))_{i \in I})$ be a trace, an execution, or a fair execution.
  Let $\sigma \in S_I$ be such that for some $n \in I$, $\sigma(i) = i$ for all $i > n$.
  Then $\gact{\sigma}{T}$ is respectively a trace, an execution, or a fair execution.
\end{proposition}

\begin{corollary}
  \label[corollary]{prop:main:6}
  Fairness is invariant under permutation, that is, if $\mc{R}$ is interference-free from $M_0$, $T$ is a fair trace from $M_0$, and $\Sigma = \gact{\sigma}{T}$ is a permutation of $T$, then $\Sigma$ is also fair.
\end{corollary}

\begin{proof}
  Let $T = (M_0, (t_i;\delta_i)_i)$ and $\delta_i = (\theta_i, \xi_i)$, and let $\Sigma$ be the trace $M_0 = \Sigma_0 \xrightarrow{(t_{\sigma(1)};\delta_{\sigma(1)})} \Sigma_1 \xrightarrow{(t_{\sigma(2)};\delta_{\sigma(2)})} \cdots$.
  Consider some rule $r \in \mc{R}$ such that $\Sigma_i \xrightarrow{(r;(\theta,\xi))} \Sigma_i'$.
  We must show that there exists a $j$ such that $\sigma(j) > \sigma(i)$, $t_{\sigma(j)}(\theta_{\sigma(j)}) \equiv r(\theta)$.

  Let the factorization $\sigma = \rho \circ \tau$ be given by \cref{lemma:main:9} for $n = \sigma(i)$.
  By \cref{prop:fair-mult-rewr:6}, we get that $\gact{\tau}{T}$ is fair.
  Moreover, by construction of $\tau$, $\gact{\tau}{T}$ and $\Sigma$ agree on the first $n$ steps and $n + 1$ multisets.
  By fairness, there exists a $k > \sigma(i)$ such that the $k$-th step in $\gact{\tau}{T}$ is $r(\theta)$.
  By construction of $\rho$, $\rho(k) > \sigma(i)$, so this step appears after $\Sigma_i$ in $\Sigma$ as desired.
  We conclude that $\Sigma$ is fair.
\end{proof}

\Cref{prop:main:6} established that permutations preserve fairness.
Relatedly, all fair traces from a given multiset are permutations of each other.
To do show this, we construct a potentially infinite sequence of permutations and use the following lemma to compose them:

\begin{lemma}
  \label[lemma]{lemma:main:2}
  Let $(\sigma_n)_{n \in I}$ be a family of bijections on $I$ such that for all $m < n$,
  \[
    (\sigma_n \circ \dotsb \circ \sigma_1)(m) = (\sigma_m \circ \dotsb \circ \sigma_1)(m).
  \]
  Let $\sigma : I \to I$ be given by $\sigma(m) = (\sigma_m \circ \dotsb \circ \sigma_1)(m)$.
  Then $\sigma$ is injective, but need not be surjective.
\end{lemma}

\begin{lemma}
  \label[lemma]{lemma:main:1}
  Let $\mc{R}$ be interference-free from $M_0$.
  Consider a fair execution $T = (M_0, (r_i; (\theta_i, \xi_i))_{i \in I})$ and a step $M_0 \xrightarrow{(t;(\tau,\rho))} M_1'$.
  Set $n = \phi_T(0,t,\tau)$ (so $t(\tau) \equiv r_n(\theta_n)$).
  Then $\gact{(1,\dotsc,n)}{T}$ is a permutation of $T$ with $(t;(\tau,\xi_n))$ as its first step, and it is a fair execution.
\end{lemma}

\begin{proposition}
  \label[proposition]{prop:main:1}
  If $\mc{R}$ is interference-free from $M_0$, then all fair executions from $M_0$ are permutations of each other.
\end{proposition}

\begin{proof}[Proof (Sketch)]
  Consider traces $R = (R_0, (r_i; (\theta_i, \xi_i))_{i \in I})$ and $T = (T_0, (t_j; (\tau_j, \zeta_j))_{j \in J})$ where $R_0 = M_0 = T_0$.
  We construct a sequence of permutations $\sigma_0, \sigma_1, \dotsc$, where $\Phi_0 = R$ and the step $\Phi_{n + 1} = \gact{\sigma_{n + 1}}{\Phi_n}$ is given by \cref{lemma:main:1} such that $\Phi_{n + 1}$ agrees with $T$ on the first $n + 1$ steps.
  We then assemble these permutations $\sigma_n$ into an injection $\sigma$ using \cref{lemma:main:2}; fairness ensures that it is a surjection.
  We have $T = \gact{\sigma}{R}$ by construction.
\end{proof}

Let the support of a multiset $M = (S, m)$ be the set $\supp(M) = \{ s \in S \mid m(s) > 0 \}$.
We say that two traces $T = (M_0;(r_i,\delta_i)_I)$ and $T'$ are \defin{union-equivalent} if $T'$ can be refreshed to a trace $(N_0;(s_j,\rho_j)_j)$ such that the unions of the supports of the multisets in the traces are equal, \ie, such that
\[
  \bigcup_{i \geq 0} \supp(M_i) = \bigcup_{j \geq 0} \supp(N_j)
\]

\begin{lemma}
  \label[lemma]{lemma:obssem:6}
  Consider an MRS and assume $T$ is a permutation of $S$.
  Then $T$ and $S$ are union-equivalent.
\end{lemma}

\begin{proof}
  Consider a trace $(M_0,(r_i;\delta_i)_i)$.
  For all $n$, each judgment in $M_n$ appears either in $M_0$ or in the result of some rule $r_i$ with $i \leq n$.
  Traces $T$ and $S$ start from the same multiset and have the same rules.
  It follows that they are union-equivalent.
\end{proof}

\Cref{prop:obssem:11} will be key in \cref{sec:observ-comm} to showing that processes have unique observations.

\begin{corollary}
  \label[corollary]{prop:obssem:11}
  If $\mc{R}$ is interference-free from $M$, then all fair executions from $M$ are union-equivalent.
\end{corollary}

\section{Session-Typed Processes}
\label{sec:sess-typed-proc}

Session types specify communication protocols between communicating processes.
In this section, we present a session-typed language arising from a proofs-as-programs interpretation of intuitionistic linear logic~\cite{caires_pfenning_2010:_session_types_intuit_linear_propos} extended to support recursive processes and recursive types.

We let $A,B,C$ range over session types and $a,b,c$ range over channel names.
A process $P$ provides a distinguished service $A_0$ over some channel $c_0$, and may use zero or more services $A_i$ on channels $c_i$.
In this sense, a process $P$ is a server for the service $A_0$, and a client of the services $A_i$.
The channels $c_1 : A_1, \dotsc, c_n : A_n$ form a linear context $\Delta$.
We write $\Delta \vdash P :: c_0 : A_0$ to capture these data.
We also allow $P$ to depend on process variables $p_i$ of type $\Tproc{b:B}{\Delta}$.
Values of type $\Tproc{b:B}{\Delta}$ are processes $Q$ such that $\Delta \vdash Q :: b : B$.
We write $\Pi$ for structural contexts of process variables $p_i : \Tproc{a_i:A_i}{\Delta_i}$.
These data are captured by the judgment $\jtypem{\Pi}{\Delta}{P}{c_0}{A_0}$, and we say that $P$ is closed if $\Pi$ is empty.

At any given point in a computation, communication flows in a single direction on a channel $c : A$.
The direction of communication is determined by the \textit{polarity} of the type $A$, where session types are partitioned as positive or negative~\cite{pfenning_griffith_2015:_polar_subst_session_types}.
Consider a process judgment $\jtypem{\Pi}{\Delta}{P}{c_0}{A_0}$.
Communication on positively-typed channels flows from left-to-right in this judgment: if $A_0$ is positive, then $P$ can only send output on $c_0$, while if $A_i$ is positive for $1 \leq i \leq n$, then $P$ can only receive input on $c_i$.
Symmetrically, communication on negatively-typed channels flows from right-to-left in the judgment.
Bidirectional communication arises from the fact that the type of a channel evolves over the course of a computation, sometimes becoming positive, sometimes becoming negative.

Most session types have a polar dual, where the direction of communication is reversed.
With one exception, we only consider positive session types here.
Negative session types pose no difficulty and can be added by dualizing the constructions.
To illustrate this dualization, we also consider the (negative) external choice type $\Tamp {\{l:A_l\}}_{l \in L}$, the polar dual of the (positive) internal choice type $\Tplus \{l:A_l\}_{l \in L}$.

The operational behaviour of closed processes is given by a substructural operational semantics (SSOS) in the form of a multiset rewriting system.
The judgment $\jproc{c}{P}$ means that the closed process $P$ provides a channel $c$.
The judgment $\jmsg{c}{m}$ means the channel $c$ is carrying a message $m$.
Process communication is asynchronous: processes send messages without synchronizing with recipients.
To ensure that messages on a given channel are received in order, the $\jmsg{c}{m}$ judgment encodes a queue-like structure similar to the queues of \cref{ex:fair-mult-rewr:1}, and we ensure that each channel name $c$ is associated with at most one $\jmsg{c}{m}$ judgment.
For example, the multiset $\jmsg{c_0}{m_0; c_0 \leftarrow c_1}, \jmsg{c_1}{m_0; c_1 \leftarrow c_2}, \dotsc$ captures the queue of messages $m_0,m_1,\dotsc$ on $c_0$.
There is no global ordering on sent messages: messages sent on different channels can be received out of order.
We extend the usual SSOS with a new persistent judgment, $\jctype{c}{A}$, which means that channel $c$ has type $A$.

The \defin{initial configuration} of $\jtypem{\cdot}{c_1:A_1,\dotsc,c_n:A_n}{P}{c_0}{A_0}$ is the multiset
\[
  \jproc{c_0}{P}, \jctype{c_0}{A_0}, \dotsc, \jctype{c_n}{A_n}.
\]
A \defin{process trace} is a trace from the initial configuration of a process, and a multiset in it is a \defin{configuration}.
A \defin{fair execution} of ${\jtypem{\cdot}{\Delta}{P}{c}{A}}$ is a fair execution from its initial configuration.

We give the typing rules and the substructural operational semantics in \cref{sec:statics-dynamics}.
In \cref{sec:properties-traces}, we study properties of process traces and fair executions.
In particular, we show that each step in these traces preserves various invariants, that the MRS of \cref{sec:statics-dynamics} is non-interfering from initial process configurations, and that every process has a fair execution.

\subsection{Statics and Dynamics}
\label{sec:statics-dynamics}

The process $\tFwdP{a}{b}$ forwards all messages from the channel $a$ to the channel $b$; it assumes that both channel have the same positive type.
It is formed by \rn{Fwd${}^+$} and its operational behaviour is given by \eqref{eq:obssem:5}.
\begin{gather}
  \infer[\rn{Fwd${}^+$}]{
    \jtypem{\Pi}{a:A}{\tFwdP{a}{b}}{b}{A}
  }{}
  \nonumber\\
  \label{eq:obssem:5}
  \jstep{
    \jmsg{a}{m},
    \jproc{b}{\tFwdP{a}{b}}
  }{
    \jmsg{b}{m}
  }
\end{gather}

Process composition $\tCut{a:A}{P}{Q}$ spawns processes $P$ and $Q$ that communicate over a shared private channel $a$ of type $A$.
It captures Milner's ``parallel composition plus hiding'' operation~\cite[pp.~20f.]{milner_1980:_calcul_commun_system}.
To ensure that the shared channel is truly private, we generate a globally fresh channel $b$ for $P$ and $Q$ to communicate over.
\begin{gather}
  \infer[\rn{Cut}]{
    \jtypem{\Pi}{\Delta_1,\Delta_2}{\tCut{a:A}{P}{Q}}{c}{C}
  }{
    \jtypem{\Pi}{\Delta_1}{P}{a}{A}
    &
    \jtypem{\Pi}{a:A,\Delta_2}{Q}{c}{C}
  }
  \nonumber\\
  \label{eq:76}
  \jstep{
    \jproc{c}{\tCut{a:A}{P}{Q}}
  }{
    \exists b.
    \jproc{b}{[b/a]P},
    \jproc{c}{[b/a]Q},
    \jctype{b}{A}
  }
\end{gather}

The process $\tClose a$ closes a channel $a$ of type $\Tu$ by sending the ``close message'' $\ast$ over $a$.
Dually, the process $\tWait{a}{P}$ blocks until it receives the close message on the channel $a$, and then continues as $P$.
\begin{gather}
  \infer[\rn{$\Tu R$}]{
    \jtypem{\Pi}{\cdot}{\tClose a}{a}{\Tu}
  }{}
  \qquad
  \infer[\rn{$\Tu L$}]{
    \jtypem{\Pi}{\Delta, a : \Tu}{\tWait{a}{P}}{c}{C}
  }{
    \jtypem{\Pi}{\Delta}{P}{c}{C}
  }
  \nonumber\\
  \label{eq:obssem:1}
  \jstep{
    \jproc{a}{\tClose{a}}
  }{
    \jmsg{a}{\ast}
  }
  \\
  \label{eq:obssem:3}
  \jstep{
    \jmsg{a}{\ast},
    \jproc{c}{\tWait{a}{P}}
  }{
    \jproc{c}{P}
  }
\end{gather}

Processes can communicate channels over channels of type $B \Tot A$, where the transmitted channel has type $B$ and subsequent communication has type $A$.
The process $\tSendC{a}{b}{P}$ sends a channel $b$ over channel $a$ and then continues as $P$.
To ensure a queue-like structure for messages on $a$, we generate a fresh channel name $d$ for the ``continuation channel'' that will carry subsequent communications.
The process $\tRecvC{b}{a}{P}$ blocks until it receives a channel over $a$, binds it to the name $b$, and continues as $P$.
Operationally, we rename $a$ in $P$ to the continuation channel $d$ carrying the remainder of the communications.
\begin{gather}
  \infer[\rn{$\Tot R^*$}]{
    \jtypem{\Pi}{\Delta, b : B}{\tSendC{a}{b}{P}}{a}{B \Tot A}
  }{
    \jtypem{\Pi}{\Delta}{P}{a}{A}
  }
  \qquad
  \infer[\rn{$\Tot L$}]{
    \jtypem{\Pi}{\Delta, a : B \Tot A}{\tRecvC{b}{a}{P}}{c}{C}
  }{
    \jtypem{\Pi}{\Delta, a : A, b : B}{P}{c}{C}
  }
  \nonumber\\
  \label{eq:main:1}
  \jstep{
    \jproc{a}{\tSendC{a}{b}{P}},
    \jctype{a}{B \Tot A}
  }{
    \exists d.
    \jproc{d}{[d/a]P},
    \jmsg{a}{\mSendC{a}{b}{d}},
    \jctype{d}{A}
  }
  \\
  \label{eq:obssem:4}
  \jstep{
    \jmsg{a}{\mSendC{a}{e}{d}},
    \jproc{c}{\tRecvC{b}{a}{P}}%
  }{
    \jproc{c}{[e,d/b,a]Q}%
  }
\end{gather}

The internal choice type $\Tplus\{l:A_l\}_{l \in L}$ offers a choice of services $A_l$.
The process $\tSendL{a}{k}{P}$ sends a label $k$ on $a$ to signal its choice to provide the service $A_k$ on $a$.
The process $\tCase{a}{\left\{l \Rightarrow P_l\right\}_{l \in L}}$ blocks until it receives a label $k \in L$ on $a$, and then continues as $P_k$.
\begin{gather}
  \infer[\rn{$\Tplus R_k$}]{
    \jtypem{\Pi}{\Delta}{\tSendL{a}{k}{P}}{a}{{\Tplus\{l:A_l\}}_{l \in L}}
  }{
    \jtypem{\Pi}{\Delta}{P}{a}{A_k}\quad(k \in L)
  }
  \qquad
  \infer[\rn{$\Tplus L$}]{
    \jtypem{\Pi}{\Delta,a:{\Tplus\{l : A_l\}}_{l \in L}}{\tCase{a}{\left\{l \Rightarrow P_l\right\}_{l \in L}}}{c}{C}
  }{
    \jtypem{\Pi}{\Delta,a:A_l}{P_l}{c}{C}\quad(\forall l \in L)
  }
  \nonumber\\
  \label{eq:74}
  \jstep{
    \jproc{a}{\tSendL{a}{k}{P}},
    \jctype{a}{\Tplus\{l:A_l\}_{l \in L}}
  }{
    \exists d.
    \jmsg{a}{\mSendL{a}{k}{d}},
    \jproc{d}{[d/a]P},
    \jctype{d}{A_k}
  }
  \\
  \label{eq:obssem:2}
  \jstep{
    \jmsg{a}{\mSendL{a}{k}{d}},
    \jproc{c}{\tCase{a}{\left\{l \Rightarrow P_l\right\}_{l \in L}}}
  }{
    \jproc{c}{[d/a]P_k}
  }
\end{gather}

To illustrate the duality between positive and negative types, we consider the (negative) external choice type.
It is the polar dual of the (positive) internal choice type.
The external choice type $\Tamp\{l:A_l\}_{l \in L}$ provides a choice of services $A_l$.
The process $\tCase{a}{\left\{l \Rightarrow P_l\right\}_{l \in L}}$ blocks until it receives a label $k \in L$ on $a$, and then continues as $P_k$.
The process $\tSendL{a}{k}{P}$ sends a label $k$ on $a$ to signal its choice to use the service $A_k$ on $a$.
Observe that, where a provider of an internal choice type \emph{sends} a label in \eqref{eq:74}, a provider of the external choice type \emph{receives} a label in \eqref{eq:sess-type-proc:2}.
Analogously, a client of an internal choice type receives \emph{receives} a label in \eqref{eq:obssem:2}, and a client of an external choice type \emph{sends} a label in \eqref{eq:sess-type-proc:1}.
\begin{gather}
  \infer[\rn{$\Tamp R$}]{
    \jtypem{\Psi}{\Delta}{\tCase{a}{\left\{l \Rightarrow P_l\right\}_{l \in L}}}{a}{{\Tamp\{l :A_l \}}_{l \in L}}
  }{
    \jtypem{\Psi}{\Delta}{P_l}{a}{A_l}\quad(\forall l \in L)
  }
  \quad
  \infer[\rn{$\Tamp L_k$}]{
    \jtypem{\Psi}{\Delta,a:{\Tamp\{l : A_l\}}_{l \in L}}{\tSendL{a}{k}{P}}{c}{C}
  }{
    \jtypem{\Psi}{\Delta,a:A_k}{P}{c}{C}
    &
    (k \in L)
  }
  \nonumber\\
  \label{eq:sess-type-proc:2}
  \jstep{
    \jmsg{a}{\mSendL{a}{k}{d}},
    \jproc{a}{\tCase{a}{\left\{l \Rightarrow P_l\right\}_{l \in L}}}
  }{
    \jproc{d}{[d/a]P_k}
  }\\
  \label{eq:sess-type-proc:1}
  \jstep{
    \jproc{c}{\tSendL{a}{k}{P}},
    \jctype{a}{\Tplus\{l:A_l\}_{l \in L}}
  }{
    \exists d.
    \jmsg{a}{\mSendL{a}{k}{d}},
    \jproc{c}{[d/a]P},
    \jctype{d}{A_k}
  }
\end{gather}

A communication of type $\Trec{\alpha}{A}$ is an unfold message followed by a communication of type $[\Trec{\alpha}{A}/\alpha]A$.
The process $\tSendU{a}{P}$ sends an unfold message and continues as $P$.
The process $\tRecvU{a}{P}$ blocks until it receives the unfold message on $a$ and continues as $P$.
\begin{gather}
  \infer[\rn{$\rho^+R$}]{
    \jtypem{\Pi}{\Delta}{\tSendU{a}{P}}{a}{\Trec{\alpha}{A}}
  }{
    \jtypem{\Pi}{\Delta}{P}{a}{[\Trec{\alpha}{A}/\alpha]A}
  }
  \qquad
  \infer[\rn{$\rho^+L$}]{
    \jtypem{\Pi}{\Delta, a : \Trec{\alpha}{A}}{\tRecvU{a}{P}}{c}{C}
  }{
    \jtypem{\Pi}{\Delta, a : [\Trec{\alpha}{A}/\alpha]A}{P}{c}{C}
  }
  \nonumber\\
  \label{eq:main:3}
  \begin{gathered}
    \jstep{
      \jproc{a}{\tSendU{a}{P}},
      \jctype{a}{\Trec{\alpha}{A}}
    }{\qquad\qquad\qquad\qquad\qquad\qquad\qquad\qquad
      \\
      \qquad\qquad
      \exists d.
      \jmsg{a}{\mSendU{a}{d}},
      \jproc{d}{[d/a]P},
      \jctype{d}{[\Trec{\alpha}{A}/\alpha]A}
    }
  \end{gathered}
  \\
  \label{eq:main:4}
  \jstep{
    \jmsg{a}{\mSendU{a}{d}},
    \jproc{c}{\tRecvU{a}{P}}
  }{
    \jproc{c}{[d/a]P}
  }
\end{gather}

Finally, recursive processes are formed in the standard way.
The SSOS is only defined on closed processes, so there are no rules for process variables.
Recursive processes step by unfolding.
\begin{gather}
  \infer[\rn{Var}]{
    \jtypem{\Pi, p : \Tproc{c:C}{\Delta}}{\Delta}{p}{c}{C}
  }{}
  \qquad
  \infer[\rn{Rec}]{
    \jtypem{\Pi}{\Delta}{\tFix{p}{P}}{c}{C}
  }{
    \jtypem{\Pi, p : \Tproc{c:C}{\Delta}}{\Delta}{P}{c}{C}
  }
  \nonumber\\
  \label{eq:main:2}
  \jstep{
    \jproc{c}{\tFix{p}{P}}
  }{
    \jproc{c}{[\tFix{p}{P}/p]P}
  }
\end{gather}

\begin{example}
  \label{conj:sess-type-proc:1}
  The protocol $\ms{conat} = \Trec{\alpha}{(\ms{z} : \Tu) \Tplus (\ms{s} : \alpha)}$ encodes conatural numbers.
  Indeed, a communication is either an infinite sequence of successor labels $\ms{s}$, or some finite number of $\ms{s}$ labels followed by the zero label $\ms{z}$ and termination.
  The following process receives a conatural number $i$ and outputs its increment on $o$:
  \[
    \jtypem{\cdot}{i : \ms{conat}}{\tSendU{o}{\tSendL{s}{o}{\tFwdP{o}{i}}}}{o}{\ms{conat}}.
  \]
  It works by outputting a successor label on $o$, and then forwarding the conatural number $i$ to $o$.
  It has the following fair execution, where we elide $\jctype{c}{A}$ judgments and annotations on the arrows:
  \begin{gather*}
    \jproc{o}{\tSendU{o}{\tSendL{s}{o}{\tFwdP{o}{i}}}}
    \longrightarrow \jmsg{c}{\mSendU{o}{o_1}}, \jproc{o_1}{\tSendL{s}{o_1}{\tFwdP{o_1}{i}}}
    \longrightarrow\\
    \jmsg{o}{\mSendU{o}{o_1}}, \jmsg{o_1}{\mSendL{s}{o_1}{o_2}}, \jproc{o_2}{\tFwd{o_2}{i}}.
  \end{gather*}
  The following recursive process outputs the infinite conatural number $s(s(s(\cdots)))$ on $o$:
  \[
    \jtypem{\cdot}{\cdot}{\tFix{\omega}{\tSendU{o}{\tSendL{s}{o}{\omega}}}}{o}{\ms{conat}}.
  \]
  It has an infinite fair execution where for $n \geq 1$, the rules $r_{3n - 2}$, $r_{3n - 1}$, and $r_{3n}$ are respectively instantiations of \eqref{eq:main:2}, \eqref{eq:main:3}, and \eqref{eq:74}.
\end{example}

\subsection{Properties of Process Traces}
\label{sec:properties-traces}

Let $\mc{P}$ be MRS given by the above rules.
We prove various invariants maintained by process traces.

Let $\fcn(P)$ be the set of free channel names in $P$.
The following result follows by an induction on $n$ and a case analysis on the rule used in the last step:

\begin{proposition}
  \label[proposition]{prop:main:3}
  Let $T = (M_0, (r_i;\delta_i)_i)$ be a process trace.
  For all $n$, if $\jproc{c_0}{P} \in M_n$, then
  \begin{enumerate}
  \item $c_0 \in \fcn(P)$;
  \item for all $c_i \in \fcn(P)$, there exists an $A_i$ such that $\jctype{c_i}{A_i} \in M_n$; and
  \item where $\fcn(P) = \{ c_0, \dotsc, c_m \}$, we have $\jtypem{\cdot}{c_1 : A_1, \dotsc, c_m : A_m}{P}{c_0}{A_0}$.
  \end{enumerate}
  If $\jmsg{c}{m} \in M_n$, then
  \begin{itemize}
  \item if $m = \jmsg{c}{\ast}$, then $\jctype{c}{\Tu} \in M_n$;
  \item if $m = \mSendL{c}{l_j}{d}$, then either $\jctype{c}{\Tplus \{ l_i : A_i \}_{i \in I}} \in M_n$ or $\jctype{c}{\Tamp \{ l_i : A_i \}_{i \in I}} \in M_n$ for some $A_i$ ($i \in I$), and $\jctype{d}{A_j} \in M_n$ for some $j \in I$.
  \item if $m = \mSendC{c}{a}{b}$, then $\jctype{c}{A \Tot B}, \jctype{a}{A}, \jctype{b}{B} \in M_n$ for some $A$ and $B$;
  \item if $m = \mSendU{c}{d}$, then $\jctype{c}{\Trec{\alpha}{A}}, \jctype{d}{[\Trec{\alpha}{A}/\alpha]A} \in M_n$ for some $\Trec{\alpha}{A}$.
  \end{itemize}
\end{proposition}

The MRS $\mc{P}$ differs from the usual MRSs given for this style session-typed languages~\cite{ gommerstadt_2018:_session_typed_concur_contr, pfenning_griffith_2015:_polar_subst_session_types, toninho_2013:_higher_order_proces_funct_session} in the addition of $\jctype{c}{A}$ judgments.
\Cref{prop:sess-type-proc:1} shows that their addition does not change the operational behaviour of the semantics.
Let $|M|$, $|\mc{P}|$, $|T|$, etc., be the result of erasing all $\jctype{c}{A}$ judgments.

\begin{corollary}
  \label[corollary]{prop:sess-type-proc:1}
  Consider a process $\jtypem{\cdot}{\Delta}{P}{c}{A}$ with initial state $M_0$.
  If $T$ is a trace from $M_0$ under $\mc{P}$, then $|T|$ is a trace from $|M_0|$ under $|\mc{P}|$.
  If $T$ is a trace from $|M_0|$ under $|\mc{P}|$, then there exists a trace $T'$ from $M_0$ under $\mc{P}$ such that $|T'| = T$.
\end{corollary}

\Cref{prop:main:3} showed that there were enough $\jctype{c}{A}$ judgments in a trace.
\Cref{prop:main:2} shows that there are not too many:

\begin{proposition}
  \label[proposition]{prop:main:2}
  Let $(M_0, (r_i;\delta_i)_i)$ be a process trace.
  For all channels $c$ and all $i, j \geq 0$, if $\jctype{c}{A_i}$ appears in $M_i$ and $\jctype{c}{A_j}$ appears in $M_j$, then $A_i = A_j$.
\end{proposition}

We show an analogous uniqueness result for $\jmsg{c}{m}$ judgments.
It implies that each channel name in an execution carries at most one message.
To prove it, we begin by partitioning a process's free channels into ``input'' and ``output'' channels and show that at all times, a channel is an output channel of at most one process.
Given a process $P$, let $\outc(P)$ be the subset of $\fcn(P)$ recursively defined by:
\begin{align*}
  \outc(\tFwdP{a}{b}) &= \{ b \} &  \outc(\tCut{a}{P}{Q}) &= (\outc(P) \cup \outc(Q)) \setminus \{ a \}\\
  \outc(\tClose{a}) &= \{ a \} & \outc(\tWait{a}{P}) &= \outc(P)\\
  \outc(\tSendL{a}{k}{P}) &= \{ a \} \cup \outc(P) & \outc(\tCase{a} (l \Rightarrow P_l)_{l \in L}) &= \left(\,\,\bigcup_{l \in L} {\outc(P_l)}\right) \setminus \{ a \}\\
  \outc(\tSendC{a}{b}{P}) &= \{ a \} \cup \outc(P) & \outc(\tRecvC{b}{a}{P}) &= \outc(P) \setminus \{ a, b \}\\
  \outc(\tSendU{a}{P}) &= \{ a \} \cup \outc(P) & \outc(\tRecvU{a}{P}) &= \outc(P) \setminus \{ a \}\\
  \outc(p) &= \emptyset & \outc(\tFix{p}{P}) &= \outc(P)
\end{align*}
Intuitively, $c \in \outc(P)$ if the next time $P$ communicates on $c$, $P$ sends a message on $c$.
Given a configuration $\mc{C}$, let $\outc(\mc{C})$ be the union of the sets $\outc(P)$ for $\jproc{c}{P}$ in $\mc{C}$.
Analogously, let $\inpc(P)$ and $\inpc(\mc{C})$ be the set of input channels of $P$ and of $\mc{C}$.

\begin{lemma}
  \label{lemma:main:5}
  If $F(\vec k) \xrightarrow{(r;(\vec k, \vec a))} G(\vec k, \vec a)$ by a rule $r$ of \cref{sec:statics-dynamics}, then
  \begin{itemize}
  \item if $\jmsg{c}{m} \in F(\vec k)$, then $c \in \inpc(F(\vec k))$;
  \item if $\jmsg{c}{m} \in G(\vec k, \vec a)$, then $c \in \outc(F(\vec k))$;
  \item if $\jmsg{c}{m; \tFwd{c}{d}} \in G(\vec k, \vec a)$, then $d \in \vec a$ and $d \in \fcn(G(\vec k, \vec a))$; and
  \item $\outc(G(\vec k, \vec a)) \subseteq \outc(F, \vec k) \cup \vec a$ and $\inpc(G(\vec k, \vec a)) \subseteq \inpc(F, \vec k) \cup \vec a$.
  \end{itemize}
\end{lemma}

\begin{proof}
  Immediate by a case analysis on the rules.
\end{proof}

An induction with \cref{lemma:main:5} implies the desired disjointness result:

\begin{lemma}
  \label[lemma]{lemma:main:6}
  Let $(M_0, (r_i;\delta_i)_i)$ be a process trace.
  For all $n$, if $\jproc{c}{P}$ and $\jproc{d}{Q}$ appear in $M_n$, then $\outc(P) \cap \outc(Q) = \emptyset$ and $\inpc(P) \cap \inpc(Q) = \emptyset$.
\end{lemma}

The following lemma shows that processes do not send messages on channels $c$ already associated with a $\jmsg{c}{m}$ judgment:

\begin{lemma}
  \label[lemma]{lemma:main:7}
  Let $(M_0, (r_i;\delta_i)_i)$ be a process trace.
  For all $n \leq k$, if $\jmsg{c}{m} \in M_n$ and $\jproc{d}{P} \in M_k$, then $c \notin \outc(P)$.
\end{lemma}

The desired result then follows by induction and the above results:

\begin{corollary}
  \label[corollary]{cor:main:1}
  Let $(M_0, (r_i;\delta_i)_i)$ be a process trace.
  For all channels $c$ and all $i, j \geq 0$, if $\jmsg{c}{m_i}$ appears in $M_i$ and $\jmsg{c}{m_j}$ appears in $M_j$, then $m_i = m_j$.
\end{corollary}

We now turn our attention to showing that all well-typed, closed processes have fair executions.
This fact will follow easily from the following proposition:

\begin{proposition}
  \label[proposition]{prop:main:8}
  The MRS $\mc{P}$ is non-overlapping from the initial configuration of $\jtypem{\cdot}{\Delta}{P}{c}{A}$ for all $\jtypem{\cdot}{\Delta}{P}{c}{A}$.
\end{proposition}

\begin{proof}
  Consider a trace $(M_0, (r_i;(\theta_i,\xi_i)))$ from the initial configuration of $\jtypem{\cdot}{\Delta}{P}{c}{A}$ and some arbitrary $n$.
  It is sufficient to show that if $s_1(\phi_1)$ and $s_2(\phi_2)$ are distinct instantiations applicable to $M_n$, then $F_1(\phi_1)$ and $F_2(\phi_2)$ are disjoint multisets: $F_1(\phi_1) \cap F_2(\phi_2) = \emptyset$.
  Indeed, if this is the case and $s_1(\phi_1), \dotsc, s_k(\phi_k)$ are the distinct rule instantiations applications to $M_n$, then $F_1(\theta_1),\dotsc,F_k(\phi_k) \subseteq M_n$, so $\olap{M_n}{F_1(\phi_1),\dotsc,F_k(\phi_k)} = \emptyset$.

  We proceed by case analysis on the possible judgments in $F_1(\phi_1) \cap F_2(\phi_2)$.
  \begin{description}
  \item[Case $\jmsg{c}{m}$.] Then $c \in \inpc(F_1(\phi_1))$ and $c \in \inpc(F_2(\phi_2))$ by \cref{lemma:main:5}.
    This is a contradiction by \cref{lemma:main:6}.
  \item[Case $\jproc{c}{P}$.] Then $s_1 = s_2$ by a case analysis on the rules.
    We show that $\phi_1 = \phi_2$.
    If $s_1$ is one of \labelcref{eq:obssem:1,eq:obssem:3,eq:obssem:5,eq:74,eq:76,eq:main:1,eq:main:3,eq:main:2,eq:sess-type-proc:1}, then we have $\phi_1 = \phi_2$, because all constants matched by $\phi_1$ and $\phi_2$ appear in $\jproc{c}{P}$.
    If $s_1$ is one of \labelcref{eq:obssem:2,eq:obssem:4,eq:main:4,eq:sess-type-proc:2}, then $F_i(\phi_i)$ contain a judgment $\jmsg{d}{m_i}$ where there is a constant $e_i \in m_i$ that appears in $\phi_i$, but not in $\jproc{c}{P}$ (explicitly, $e_i$ is the name of the continuation channel).
    By \cref{cor:main:1}, $m_1 = m_2$, so $e_1 = e_2$.
    All other channel names in $\phi_i$ appear in $\jproc{c}{P}$, so $\phi_1 = \phi_2$.
    So $s_1(\phi_1)$ and $s_2(\phi_2)$ are not distinct rule instantiations, a contradiction.
  \item[Case $\jctype{c}{A}$.] By case analysis on the rules, $s_1 = s_2$ and there exist judgments $\jproc{d_i}{P_i} \in F_i(\phi_i)$.
    Suppose to the contrary that $P_1 \neq P_2$.
    By case analysis on the rules, $s_1$ is one of \labelcref{eq:74,eq:main:1,eq:sess-type-proc:1,eq:main:3}.
    This implies that $c \in \outc(P_1) \cap \outc(P_2)$, a contradiction of \cref{lemma:main:6}.
    So $P_1 = P_2$.
    Because all constants in $\phi_1$ and $\phi_2$ appear in $P_1$, we conclude that $\phi_1 = \phi_2$.
    So $s_1(\phi_1)$ and $s_2(\phi_2)$ are not distinct rule instantiations, a contradiction.\qedhere
  \end{description}
\end{proof}

\begin{corollary}
  \label[corollary]{cor:main:3}
  Every process $\jtypem{\cdot}{\Delta}{P}{c}{A}$ has a fair execution.
  Its fair executions are all permutations of each other and they are all union-equivalent.
\end{corollary}

\begin{proof}
  By \cref{prop:main:8}, $\mc{P}$ is non-overlapping from the initial configuration $M_0$ of $\jtypem{\cdot}{\Delta}{P}{c}{A}$.
  It is then interference-free from $M_0$ by \cref{prop:main:13}, so a fair execution exists by \cref{prop:main:9}.
  All of its fair executions are permutations of each other by \cref{prop:main:1}.
  They are union-equivalent by \cref{prop:obssem:11}.
\end{proof}

\section{Observed Communications}
\label{sec:observ-comm}

Consider a closed process $\jtypem{\cdot}{c_1 : A_1,\dotsc,c_n : A_n}{P}{c_0}{A_0}$.
In this section, we will define the observation of $P$ to be a tuple $(c_i : v_i)_{0 \leq i \leq n}$, where $v_i$ is the communication of type $A_i$ observed on channel $c_i$ in a fair execution of $P$.
We extract communications from fair executions using a coinductively defined judgment.
We colour-code the modes of judgments, where \inm{inputs} to a judgment are in blue and \outm{outputs} are in red.

We begin by defining session-typed communications.
Let a communication $v$ be a (potentially infinite) tree generated by the following grammar, where $k$ and $l_i$ range over labels.
We explain these communications $v$ below when we associate them with session types.
For convenience, we also give a grammar generating the session types $A$ of \cref{sec:statics-dynamics}.
Session types are always finite expressions, and we treat $\Trec{\alpha}{A}$ as a binding operator.
\begin{align*}
  v,v' &\coloneqq \bot_A \mid \ast \mid (k, v) \mid (v, v') \mid (\unfold, v)\\
  A,A_i,B &\coloneqq \alpha \mid \Tu \mid A \Tot B \mid \Tplus ( l_1 : A_1, \dotsc, l_n : A_n ) \mid \Tamp ( l_1 : A_1, \dotsc, l_n : A_n ) \mid \Trec{\alpha}{A}.
\end{align*}
As in \cref{sec:statics-dynamics}, we abbreviate $\Tplus ( l_1 : A_1, \dotsc, l_n : A_n )$ and $\Tamp ( l_1 : A_1, \dotsc, l_n : A_n )$ by $\Tplus \{ l : A_l \}_{l \in L}$ and $\Tamp \{ l : A_l \}_{l \in L}$, respectively, where $L$ is the finite set of labels.

Next, we associate communications with session types.
The judgment $\jsynt{v}{A}$ means that the syntactic communication $v$ has type $A$.
It is coinductively defined by the following rules, where $A$ is assumed to have no unbound occurrences of $\alpha$.
The rule forming $\jsynt{(k,v_k)}{\Tplus \{ l : A_l \}_{l \in L}}$ has the side condition $k \in L$.
\[
  \infer{\jsynt{\bot_\Tu}{\Tu}}{}
  \qquad
  \infer{\jsynt{\ast}{\Tu}}{}
  \qquad
  \infer{\jsynt{\bot_{A \Tot B}}{A \Tot B}}{}
  \qquad
  \infer{
    \jsynt{(v,v')}{A \Tot B}
  }{
    \jsynt{v}{A}
    &
    \jsynt{v'}{B}
  }
  \qquad
  \infer{
    \jsynt{\bot_{\Trec{\alpha}{A}}}{\Trec{\alpha}{A}}
  }{}
  \qquad
  \infer{
    \jsynt{(\unfold, v)}{\Trec{\alpha}{A}}
  }{
    \jsynt{v}{[\Trec{\alpha}{A}/\alpha]A}
  }
\]
\[
  \infer{
    \jsynt{\bot_{\Tplus \{ l : A_l \}_{l \in L}}}{\Tplus \{ l : A_l \}_{l \in L}}
  }{}
  \quad
  \infer{
    \jsynt{(k,v_k)}{\Tplus \{ l : A_l \}_{l \in L}}
  }{
    \jsynt{v_k}{A_k}
  }
  \quad
  \infer{
    \jsynt{\bot_{\Tamp \{ l : A_l \}_{l \in L}}}{\Tamp \{ l : A_l \}_{l \in L}}
  }{}
  \quad
  \infer{
    \jsynt{(k,v_k)}{\Tamp \{ l : A_l \}_{l \in L}}
  }{
    \jsynt{v_k}{A_k}
  }
\]
Every closed session type $A$ has an empty communication $\bot_A$ representing the absence of communication of that type.
The communication $\ast$ represents the close message.
A communication of type $\Tplus \{ l : A_l \}_{l \in L}$ or $\Tplus \{ l : A_l \}_{l \in L}$ is a label $k \in L$ followed by a communication $v_k$ of type $A_k$, whence the communication $(k, v_k)$.
Though by itself the communication $(k, v_k)$ does not capture the direction in which the label $k$ travelled, this poses no problem to our development: we never consider communications without an associated session type, and the polarity of the type specifies the direction in which $k$ travels.
We cannot directly observe channels, but we can observe communications over channels.
Consequently, we observe a communication of type $A \Tot B$ as a pair $(v, v')$ of communications $v$ of type $A$ and $v'$ of type $B$.
A communication of type $\Trec{\alpha}{A}$ is an unfold message followed by a communication of type $[\Trec{\alpha}{A}/\alpha]A$.

Given a trace $T = (M_0, (r_i;(\theta_i,\xi_i))_{i})$, we write $\mc{T}$ for the set-theoretic union of the $M_i$, that is, $x \in \mc{T}$ if and only if $x \in \supp(M_i)$ for some $i$.
Write $\jttp{T}{c}{A}$ if $\jctype{c}{A} \in \mc{T}$.
This judgment is defined on all channel names $c$ that appear in $T$ by \cref{prop:main:3} and it is a function by \cref{prop:main:2}.

Assuming the channel $c$ appears in $T$, the judgment $\jtoc{T}{v}{c}{A}$ means that we observed a communication $v$ of type $A$ on the channel $c$ during $T$.
We will show below that whenever $\jtoc{T}{v}{c}{A}$, we also have $\jttp{T}{c}{A}$ and $\jsynt{v}{A}$.
Fixing $T$, the judgment $\jtoc{T}{v}{c}{A}$ is coinductively defined by the following rules, \ie, $\jtoc{T}{v}{c}{A}$ is the largest set of triples $(v,c,A)$ closed under the following rules.

We observe no communications on a channel $c$ if and only if $\jmsg{c}{m}$ does not appear in the trace for any $m$.
Subject to the side condition that for all $m$, $\jmsg{c}{m} \notin \mc{T}$, we have the rule
\[
  \infer[\rn{O-$\bot$}]{
    \jtoc{T}{\bot_A}{c}{A}
  }{
    \jttp{T}{c}{A}
  }
\]
We observe a close message on $c$ if and only if the close message was sent on $c$:
\[
  \infer[\rn{O-$\Tu$}]{
    \jtoc{T}{\ast}{c}{\Tu}
  }{
    \jmsg{c}{\ast} \in \mc{T}
  }
\]
We observe label transmission as labelling communications on the continuation channel.
We rely on the judgment $\jttp{T}{c}{\Tplus \{ l : A_l \}_{l \in L}}$ or $\jttp{T}{c}{\Tamp \{ l : A_l \}_{l \in L}}$ to determine the type of $c$:
\[
  \infer[\rn{O-$\Tplus$}]{
    \jtoc{T}{(l,v)}{c}{\Tplus \{ l : A_l \}_{l \in L}}
  }{
    \jmsg{c}{\mSendL{c}{l}{d}} \in \mc{T}
    &
    \jtoc{T}{v}{d}{A_l}
    &
    \jttp{T}{c}{\Tplus \{ l : A_l \}_{l \in L}}
  }
\]
\[
  \infer[\rn{O-$\Tamp$}]{
    \jtoc{T}{(l,v)}{c}{\Tamp \{ l : A_l \}_{l \in L}}
  }{
    \jmsg{c}{\mSendL{c}{l}{d}} \in \mc{T}
    &
    \jtoc{T}{v}{d}{A_l}
    &
    \jttp{T}{c}{\Tamp \{ l : A_l \}_{l \in L}}
  }
\]
As described above, we observe channel transmission as pairing of communications:
\[
  \infer[\rn{O-$\Tot$}]{
    \jtoc{T}{(u,v)}{c}{A \Tot B}
  }{
    \jmsg{c}{\mSendC{c}{a}{d}} \in \mc{T}
    &
    \jtoc{T}{u}{a}{A}
    &
    \jtoc{T}{v}{d}{B}
  }
\]
Finally, we observe the unfold message as an unfold message:
\[
  \infer[\rn{O-$\rho$}]{
    \jtoc{T}{(\unfold, v)}{c}{\Trec{\alpha}{A}}
  }{
    \jmsg{c}{\mSendU{c}{d}} \in \mc{T}
    &
    \jtoc{T}{v}{d}{[\Trec{\alpha}{A}/\alpha]A}
  }
\]

The following three propositions imply that for any $T$, $\jtoc{T}{v}{c}{A}$ is a total function from channel names $c$ in $T$ to session-typed communications $\jsynt{v}{A}$.

\begin{proposition}
  \label[proposition]{prop:main:12}
  If $\jtoc{T}{v}{c}{A}$, then $\jsynt{v}{A}$.
\end{proposition}

\begin{proof}
  Immediate by rule coinduction.
\end{proof}

\begin{proposition}
  \label[proposition]{prop:main:11}
  If $T$ is a process trace, then for all $c$, if $\jttp{T}{c}{A}$, then $\jtoc{T}{v}{c}{A}$ for some $v$.
\end{proposition}

\begin{proof}[Proof (Sketch)]
  Let $S$ be the set of all triples $(v, A, c)$ for session-typed communications $\jsynt{v}{A}$ and channel names $c$.
  Let $\Phi : \wp(S) \to \wp(S)$ be the rule functional defining $\jtoc{T}{v}{c}{A}$.
  Then the judgment $\jtoc{T}{v}{c}{A}$ is given by the greatest fixed point $\gfp(\Phi)$ of $\Phi$ on the complete lattice $\wp(S)$, where $\jtoc{T}{v}{c}{A}$ if and only if $(v, A, c) \in \gfp(\Phi)$.
  The functional $\Phi$ is cocontinuous by \cite[Theorem~2.9.4]{sangiorgi_2012:_introd_bisim_coind}, so $\gfp(\Phi) = \bigcap_{n \geq 0} \Phi^n(S)$ by \cite[Theorem~2.8.5]{sangiorgi_2012:_introd_bisim_coind}.
  It is sufficient to show that if $\jttp{T}{c}{A}$, then there exists a $v$ such that $(c,v,A) \in \Phi^n(S)$ for all $n$.
  This $v$ can be constructed using a coinductive argument and a case analysis on $\jmsg{c}{m} \in \mc{T}$.
\end{proof}

\begin{proposition}
  \label[proposition]{prop:obssem:4}
  If $T$ is a trace from the initial configuration of a process, then for all $c$, if $\jtoc{T}{v}{c}{A}$ and $\jtoc{T}{w}{c}{B}$, then $v = w$ and $A = B$.
\end{proposition}

\begin{proof}[Proof (Sketch).]
  Let $R = \{ (\jtoc{T}{v}{c}{A}, \jtoc{T}{w}{c}{B}) \mid \exists v, w, c, A, B . \jtoc{T}{v}{c}{A} \land \jtoc{T}{w}{c}{B} \}$.
  We claim that $R$ is a bisimulation.
  Indeed, let $(\jtoc{T}{v}{c}{A}, \jtoc{T}{w}{c}{B}) \in R$ be arbitrary.
  By \cref{cor:main:1}, at most one rule is applicable to form a judgment of the form $\jtoc{T}{u}{c}{C}$ (with $c$ fixed), so $\jtoc{T}{v}{c}{A}$ and $\jtoc{T}{w}{c}{B}$ were both formed by the same rule.
  A case analysis shows on this rule shows that $R$ satisfies the definition of a bisimulation.

  Consider arbitrary $\jtoc{T}{v}{c}{A}$ and $\jtoc{T}{w}{c}{B}$.
  They are related by $R$, so they are bisimilar.
  By \cite[Theorem~2.7.2]{jacobs_rutten_2012:_introd_coalg_coind}, bisimilar elements of the terminal coalgebra are equal, so $v = w$ and $A = B$.
\end{proof}

\Cref{prop:main:10} gives the converse of \cref{prop:main:11}:

\begin{corollary}
  \label[corollary]{prop:main:10}
  If $T$ is a process trace, then for all $c$, if $\jtoc{T}{v}{c}{A}$, then $\jttp{T}{c}{A}$.
\end{corollary}

\begin{proof}
  We show by case analysis on the rules that if $\jtoc{T}{v}{c}{A}$, then $\jttp{T}{c}{B}$ for some $B$.
  The case \rn{O-$\bot$} is obvious, while for each other case, if $\jtoc{T}{v}{c}{A}$, then $\jmsg{c}{m} \in \mc{T}$ for some $m$.
  For each of these cases, \cref{prop:main:3} implies $\jctype{c}{B} \in \mc{T}$ for some $B$, \ie, $\jttp{T}{c}{B}$.

  Assume $\jtoc{T}{v}{c}{A}$.
  By the claim, $\jttp{T}{c}{B}$ for some $B$.
  By \cref{prop:main:11}, there exists a $w$ such that $\jtoc{T}{w}{c}{B}$.
  By \cref{prop:obssem:4}, $A = B$, so $\jttp{T}{c}{A}$.
\end{proof}

\begin{theorem}
  \label[theorem]{theorem:main:2}
  Let $T$ be a fair execution of $\jtypem{\cdot}{c_1:A_1,\dotsc,c_n:A_n}{P}{c_0}{A_0}$.
  For all $0 \leq i \leq n$, there exist unique $v_i$ such that $\jsynt{v_i}{A_i}$ and $\jtoc{T}{v_i}{c_i}{A_i}$.
\end{theorem}

\begin{proof}
  By definition of fair execution, we have $\jctype{c_i}{A_i} \in \mc{T}$ for all $0 \leq i \leq n$, \ie, $\jttp{T}{c_i}{A_i}$ for all $0 \leq i \leq n$.
  By \cref{prop:main:11}, for all $0 \leq i \leq n$, there exists a $v_i$ such that $\jtoc{T}{v_i}{c_i}{A_i}$, and $\jsynt{v_i}{A_i}$ by \cref{prop:main:12}.
  Each $v_i$ is unique by \cref{prop:obssem:4}.
\end{proof}

The following theorem captures the confluence property typically enjoyed by SILL-style languages:

\begin{theorem}
  \label[theorem]{theorem:main:1}
  Let $T$ and $T'$ be a fair executions of $\jtypem{\cdot}{c_1:A_1,\dotsc,c_n:A_n}{P}{c_0}{A_0}$.
  For all $0 \leq i \leq n$, if $\jtoc{T}{v_i}{c_i}{A_i}$ and $\jtoc{T'}{w_i}{c_i}{A_i}$, then $v_i = w_i$.
\end{theorem}

\begin{proof}
  Assume $\jtoc{T}{v_i}{c_i}{A_i}$ and $\jtoc{T'}{w_i}{c_i}{A_i}$.
  By \cref{cor:main:3}, traces $T$ and $T'$ are union-equivalent, \ie, $\mc{T} = \mc{T}'$.
  It immediately follows that $\jtoc{T'}{w_i}{c_i}{A_i}$ if and only if $\jtoc{T}{w_i}{c_i}{A_i}$.
  So $v_i = w_i$ by \cref{prop:obssem:4}.
\end{proof}

We use \cref{theorem:main:1,theorem:main:2} to define the \defin{operational observation} $\obsbr{\jtypem{\cdot}{c_1:A_1,\dotsc,c_n:A_n}{P}{c_0}{A_0}}$ of $\jtypem{\cdot}{c_1:A_1,\dotsc,c_n:A_n}{P}{c_0}{A_0}$.
It is the tuple of observed communications
\[
  \obsbr{\jtypem{\cdot}{c_1:A_1,\dotsc,c_n:A_n}{P}{c_0}{A_0}} = (c_0 : v_0, \dotsc, c_n : v_n)
\]
where $\jtoc{T}{v_i}{c_i}{A_i}$ for $0 \leq i \leq n$ for some fair execution $T$ of $\jtypem{\cdot}{c_1:A_1,\dotsc,c_n:A_n}{P}{c_0}{A_0}$.
Such a $T$ exists by \cref{cor:main:3}, and $\obsbr{\jtypem{\cdot}{c_1:A_1,\dotsc,c_n:A_n}{P}{c_0}{A_0}}$ does not depend on the choice of $T$ by \cref{theorem:main:1}.
The $v_i$ such that $\jtoc{T}{v_i}{c_i}{A_i}$ exist by \cref{prop:main:2}, and they are unique by \cref{prop:obssem:4}.

Uniqueness of operational observations and \cref{theorem:main:1} crucially depend on fairness.
Indeed, without fairness a process can have infinitely many observations.
To see this, let $\Omega$ and $B$ respectively be given by
\begin{gather*}
  \jtypem{\cdot}{\cdot}{\tFix{\omega}{\omega}}{a}{\Tu}\\
  \jtypem{\cdot}{a : \Tu}{\tFix{p}{\tSendU{b}{\tSendL{b}{l}{p}}}}{b}{\Trec{\beta}{\Tplus\{l:\beta\}}}
\end{gather*}
Rule~\eqref{eq:76} is the first step of any execution of their composition $\jtypem{\cdot}{\cdot}{\tCut{a : \Tu}{\Omega}{B}}{b}{\Trec{\beta}{\Tplus\{l:\beta\}}}$.
It spawns $\Omega$ and $B$ as separate processes.
Without fairness, an execution could then consist exclusively of applications of rule \eqref{eq:main:2} to $\Omega$.
This would give the observation $(b : \bot)$.
Alternatively, $B$ could take finitely many steps, leading to observations where $b$ is a tree of correspondingly finite height.
Fairness ensures that $B$ and $\Omega$ both take infinitely many steps, leading to the unique observation $(b : (\ms{unfold}, (l, (\ms{unfold}, \dotsc))))$.

Operational observation does not take into account the order in which a process sends on channels.
For example, the following processes have the same operational observation $(a : (l, \bot_\Tu), b : (r, \bot_\Tu))$, even though they send on $a$ and on $b$ in different orders:
\begin{gather*}
  \jtypem{\cdot}{a : \Tamp \{ l : \Tu \}}{\tSendL{a}{l}{\tSendL{b}{r}{\tFwdP{a}{b}}}}{b}{\Tplus \{r : \Tu\}}\\
  \jtypem{\cdot}{a : \Tamp \{ l : \Tu \}}{\tSendL{b}{r}{\tSendL{a}{l}{\tFwdP{a}{b}}}}{b}{\Tplus \{r : \Tu\}}.
\end{gather*}
The order in which channels are used does not matter for several reasons.
First, messages are only ordered on a per-channel basis, and messages sent on different channels can arrive out of order.
Second, each channel has a unique pair of endpoints, and the \rn{Cut} rule organizes processes in a tree-like structure.
This means that two processes communicating with a process $R$ cannot at the same time also directly communicate with each other to compare the order in which $R$ sent them messages.
In other words, the ordering cannot be distinguished by other processes.

Our notion of operational observation scales to support language extensions.
Indeed, for each new session type one first defines its corresponding session-typed communications.
Then, one specifies how to observe message judgments $\jmsg{c}{m}$ in a trace as communications.
Informally, it seems desirable to ensure that if two message judgments $\jmsg{c}{m}$ can be distinguished by a receiving process, then they are observed as different session-typed communications.

A \defin{typed context} $\jtypem{\cdot}{\Delta}{C{[{\cdot}]}^{\Delta'}_{a:A}}{b}{B}$ is a context derived using the process typing rules of \cref{sec:statics-dynamics}, plus exactly one instance of the axiom
\[
  \infer[\rn{Hole}]{%
    \jtypem{\cdot}{\Delta'}{[\cdot]^{\Delta'}_{a:A}}{a}{A}
  }{}
\]
Given a context $\jtypem{\cdot}{\Delta}{C[\cdot]^{\Delta'}_{a:A}}{b}{B}$ and a process $\jtypem{\cdot}{\Delta'}{P}{a}{A}$, we let $\jtypem{\cdot}{\Delta}{C[P]}{b}{B}$ be the result of ``plugging'' $P$ into the hole, that is, of replacing the axiom \rn{Hole} by the derivation $\jtype{\Delta'}{P}{a}{A}$ in the derivation $\jtype{\Delta}{C[\cdot]^{\Delta'}_{a:A}}{b}{B}$.

We say that processes $\jtypem{\cdot}{\Delta}{P}{c}{C}$ and $\jtypem{\cdot}{\Delta}{Q}{c}{C}$ are \defin{observationally congruent}, $P \approx Q$, if $\obsbr{\jtypem{\cdot}{\Delta'}{C[P]}{b}{B}} = \obsbr{\jtypem{\cdot}{\Delta'}{C[Q]}{b}{B}}$ for all typed contexts $\jtypem{\cdot}{\Delta'}{C{[{\cdot}]}^{\Delta}_{c:C}}{b}{B}$.
Intuitively, this means that no context $C$ can differentiate processes $P$ and $Q$.

To illustrate observational congruence, we show that process composition is associative:

\begin{proposition}
  \label{prop:observ-comm:2}
  We have $\tCut{c_1 : C_1}{P_1}{(\tCut{c_2 : C_2}{P_2}{P_3})} \approx {\tCut{c_2 : C_2}{(\tCut{c_1 : C_1}{P_1}{P_2})}{P_3}}$ for all ${\jtypem{\cdot}{\Delta_1}{P_1}{c_1}{C_1}}$, all ${\jtypem{\cdot}{c_1:C_1,\Delta_2}{P_2}{c_2}{C_2}}$, and all ${\jtypem{\cdot}{c_2:C_2,\Delta_3}{P_3}{c_3}{C_3}}$.
\end{proposition}

\begin{proof}[Proof (Sketch).]
  Let $L = \tCut{c_1 : C_1}{P_1}{(\tCut{c_2 : C_2}{P_2}{P_3})}$ and $R = \tCut{c_2 : C_2}{(\tCut{c_1 : C_1}{P_1}{P_2})}{P_3}$.
  Consider some arbitrary observation context $C[\cdot]$ and a fair execution $T$ of $C[L]$.
  It is sufficient to show that $T$ agrees on message judgments with a fair execution $C[R]$.
  Union-equivalence of process traces is invariant under permutation, so we can assume without loss of generality that whenever $\jproc{c_3}{L}$ appears in some $M_n$ of $T$, then the next two steps are applications \eqref{eq:76} to decompose $L$:
  \begin{gather*}
    \jproc{c_3}{L}
    \longrightarrow \jproc{c_1'}{[c_1'/c_1]P_1},\jproc{c_3}{[c_1'/c_1](\tCut{c_2 : C_2}{P_2}{P_3})}
    \longrightarrow
    \qquad\qquad\qquad
    \\
    \jproc{c_1'}{[c_1'/c_1]P_1},\jproc{c_2'}{[c_1',c_2'/c_1,c_2]P_2},\jproc{c_3'}{[c_2'/c_2]P_3}
  \end{gather*}
  (For conciseness, we elide the $\jctype{c}{A}$ judgments.)
  There exists a fair execution $T'$ of $C[R]$ that agrees with $T$ on all steps, except for those involving $R$, where we make the same assumption:
  \begin{gather*}
    \jproc{c_3}{R}
    \longrightarrow \jproc{c_2'}{[c_2'/c_2](\tCut{c_1 : C_1}{P_1}{P_2})},\jproc{c_3}{[c_2'/c_2]P_3}
    \longrightarrow
    \qquad\qquad\qquad
    \\
    \jproc{c_1'}{[c_1'/c_1]P_1},\jproc{c_2'}{[c_1',c_2'/c_1,c_2]P_2},\jproc{c_3'}{[c_2'/c_2]P_3}
  \end{gather*}
  So traces $T$ and $T'$ agree on all message judgments, whence $\obsbr{C[L]} = \obsbr{C[R]}$.
\end{proof}

\section{Related Work}
\label{sec:related-work}

Multiset rewriting systems with existential quantification were first introduced by \textcite{cervesato_1999:_meta_notat_protoc_analy}.
They were used to study security protocols and were identified as the first-order Horn fragment of linear logic.
Since, MRSs have modelled other security protocols, and strand spaces~\cite{cervesato_2000:_inter_stran_linear_logic,cervesato_2005:_compar_between_stran}.
\Textcite{cervesato_scedrov_2009:_relat_state_based} studied the relationship between MRSs and linear logic.
These works do not explore~fairness.

Weak and strong fairness were first introduced by \textcite{apt_olderog_1982:_proof_rules_dealin_with_fairn, park_1982:_predic_trans_weak_fair_iterat} in the context of do-od languages, and were subsequently adapted to process calculi, \eg, by \textcite{costa_stirling_1987:_weak_stron_fairn_ccs} for Milner's CCS.
Our novel notion of fairness for multiset rewriting systems in \cref{sec:fair-mult-rewr} implies strong process fairness (so also weak process fairness) for the session-typed processes of \cref{sec:sess-typed-proc}.
We conjecture that this notion of fairness is stronger than required for many applications.
In future work, we intend to explore other formulations of fairness for MRSs and their impact on applications.

Substructural operational semantics~\cite{simmons_2012:_subst_logic_specif} based on multiset rewriting are widely used to specify the operational behaviour of session-typed languages arising from proofs-as-processes interpretations of linear logic and adjoint logic.
Examples include functional languages with session-typed concurrency~\cite{toninho_2013:_higher_order_proces_funct_session}, languages with run-time monitoring~\cite{gommerstadt_2018:_session_typed_concur_contr}, message-passing interpretations of adjoint logic~\cite{pruiksma_pfenning_2019:_messag_passin_inter_adjoin_logic}, and session-typed languages with sharing~\cite{balzer_pfenning_2017:_manif_sharin_with_session_types}.
The fragment of \cref{sec:statics-dynamics} illustrates some of the key ideas of this approach, and extends to these richer settings.

Some of these languages are already equipped with observational equivalences.
For example, \textcite{perez_2014:_linear_logic_relat} introduced \textit{typed context bisimilarity}, a labelled bisimilarity for session-typed processes.
It does not support recursive processes or recursive session types.
\Textcite{toninho_2015:_logic_found_session_concur_comput} explored barbed congruence for session-typed processes and shows that it coincides with logical equivalence.
\Textcite{kokke_2019:_better_late_than_never} showed that the usual notions of bisimilarity and barbed congruence carry over from the $\pi$-calculus.
They also gave a denotational semantics using Brzozowski derivatives to ``hypersequent classical processes'' that built on Atkey's denotational semantics for CP, and showed that all three notions of equivalence agreed on well-typed programs.
In future work, we intend to show that our observational congruence agrees with barbed congruence.
\Textcite{gommerstadt_2018:_session_typed_concur_contr} define a bisimulation-style observational equivalence on multisets in process traces.
It deems two configurations equivalent if whenever both configurations send an externally visible message, then the messages are equivalent.
It is easy to adapt this bisimulation to also require that one configuration sends an externally visible message if and only if the other does.
We conjecture that this modified observational equivalence coincides with the one defined in \cref{sec:observ-comm}.

Session-typed languages enjoy other notions of process equivalence.
Several session-typed languages are equipped with denotational semantics, and denotational semantics induce a compositional notion of program equivalence.
For example, \textcite{castellan_yoshida_2019:_two_sides_same_coin} gave a game semantics to a session-typed $\pi$-calculus with recursion, where session types denote event structures that encode games, and processes denote maps that encode strategies.
\Textcite{kavanagh_2020:_domain_seman_higher_v3} gave a domain-theoretic semantics to a full-featured functional language with session-typed message passing concurrency, where session types denote domains of communications and processes are continuous functions between these.

Atkey's observed communication semantics~\cite{atkey_2017:_obser_commun_seman_class_proces} for Wadler's CP~\cite{wadler_2014:_propos_as_session} was motivated by two problems.
Because CP uses a synchronous communication semantics, processes need partners to communicate with and get stuck if they try to communicate on a free channel.
On the one hand, if processes have partners, then their communication are hidden by the \rn{Cut} rule and cannot be observed, while on the other hand, if we leave the channels free, then we need to introduce reduction rules (``commuting conversions'') for stuck processes, and these rules do not correspond to operationally justified communication steps.
Atkey's elegant solution to this tension was to give communication partners to processes with free channels via closing ``configurations'', and then observing communications on these channels.
Our task in \cref{sec:observ-comm} is made easier by the fact that we use an asynchronous communication semantics.
In our setting, a process can send messages on free channels, and we can observe these without having to provide it with communication partners via configurations.
Atkey's observational equivalence and ours suffer from the same weakness: to reason about observational equivalence, we must quantify over all observation contexts.
Atkey addresses this by relating his semantics to a denotational semantics for CP and showing that they induce the same notion of equivalence.
We are actively working on relating our OCS to Kavanagh's domain semantics~\cite{kavanagh_2020:_domain_seman_higher_v3}.
Indeed, our OCS is largely motivated by efforts to relate denotational semantics of session-typed languages to their existing substructural operational semantics.
We believe that our results on fair executions and their permutations should also simplify reasoning about observational equivalence.

\section{Conclusion and Acknowledgements}
\label{sec:concl-ackn}

We studied fair executions of multiset rewriting systems, and gave various conditions for an MRS to have fair executions.
We used these results to define an observed communication semantics for session-typed languages that are defined by substructural operational semantics: the observation of a process is its communications on its free channels.
Processes are then observationally equivalent if they cannot be distinguished through communication.
We believe this work lays the foundation for future work on the semantics of session-typed processes, and in particular, we hope that it will be useful for exploring other notions of process equivalence.

The author thanks Stephen Brookes, Iliano Cervesato, Frank Pfenning, and the anonymous reviewers for their helpful comments.

\printbibliography

\end{document}